\documentclass[a4paper]{article}

\usepackage{amsfonts, amscd, amsmath, amssymb, amsthm}
\usepackage{mathtools}
\usepackage[utf8]{inputenc}
\usepackage[T1]{fontenc}

\usepackage{graphicx}
\usepackage{color}

\renewcommand{\theenumi}{{\rm\roman{enumi}}}
\renewcommand{\labelenumi}{\theenumi)}

\newcommand{\scalar}[2]{\langle#1\,,#2\rangle}

\newcommand{\A}{\mathcal{A}}
\newcommand{\ta}{\tilde{a}}
\newcommand{\te}{\tilde{e}}

\newcommand{\V}{\mathcal{V}}
\renewcommand{\H}{\mathcal{H}}

\newcommand{\CC}{\mathbb{C}}
\newcommand{\R}{\mathbb{R}}

\newcommand{\Phid}{\Phi^\dagger\negthinspace}
\makeatletter
    \@ifdefinable{\C/}{\def\C/{$C^*$-algebra}}
\makeatother

\newcommand{\map}[5][\phantom]{
\begin{array}{ccl}
\mathllap{#1{\,:\,}}{#2}&\to&{#3}\\
          {#4}&\mapsto&{#5}
\end{array}
}

\newtheorem{theorem}{Theorem}

\newtheorem{proposition}[theorem]{Proposition}

\theoremstyle{definition}
\newtheorem{definition}[theorem]{Definition}

\theoremstyle{remark}

\numberwithin{equation}{section}
\numberwithin{theorem}{section}

\begin{document}

\title{On Uhlmann's proof of the Monotonicity of the Relative Entropy}

\author{Juan Manuel Pérez-Pardo$^{1,2}$}

\date{}

\maketitle

\begin{enumerate}
\def\oldenumi{\labelenumi}
\renewcommand{\theenumi}{\arabic{enumi}}
\renewcommand{\labelenumi}{$^{\theenumi}$}
\item Universidad Carlos III de Madrid\\
Avda. de la Universidad 30, 28911 Leganés (Madrid), Spain.
\item Instituto de Ciencias Matemáticas (CSIC - UAM - UC3M - UCM)\\ 
Nicolás Cabrera, 13-15, 28049 Cantoblanco (Madrid), Spain.
\def\labelenumi{\oldenumi}
\end{enumerate}

\begin{abstract}
This article presents in a self-contained way A.~Uhlmann's celebrated Theorem of monotonicity of the relative entropy under completely positive and trace preserving maps. The Theorem is presented in its more general form and meaningful examples are given.
\end{abstract}




\section{Previous Remarks}

The main results presented in this article are the construction of the functional calculus of positive Hemitian quadratic forms developed by W.~Pusz and S.L.~Woronowicz and A.~Uhlmann's celebrated Theorem of monotoni\-city of the relative entropy on a \C/ $\A$ under the action of a positive and unital map, i.e. with trace preserving dual, $\Phi:\tilde{\A} \to \A$.\\

My interest in this topic originated during discussions with Balachandran around the years 2018/2019 on the properties of monotonicity of the relative entropy  when considering the subalgebras generated by observables of a quantum subsystem and in connection with previous results obtained by him and coworkers\cite{bal13a,bal13b,bal13c}.

The results presented in this article can be found in the original articles by  W.~Pusz and S.L.~Woronowicz\cite{pusz75} and by A.~Uhlmann\cite{uhlmann77}. A.~Uhlmann's original proof is often perceived as difficult to follow. In my opinion this is so because most of the technical results and constructions depend on the construction of the functional calculus developed by Pusz and Woronowicz. This is the reason why I originally prepared the set of notes that have served as the core of this article and that I have shared many times since. On this occasion I decided to revise the notes and give them a more polished format in the hope that they can be helpful for researchers on these topics in the future.


\section{Functional Calculus of positive quadratic forms}\label{sec:functional_calculus}

In order to proof the monotonicity of the relative entropy for completely positive and trace preserving maps we are going to take Uhlmann's definition of Relative Entropy \cite{uhlmann77}. This is based on interpolations of quadratic forms and uses their properties intensively. The interpolation of quadratic forms relies on a remarkable construction due to W.~Pusz and S.L.~Woronowicz \cite{pusz75} that allows for the definition of a functional calculus of quadratic forms defined on any vector space $\V$. For our purposes this vector space will be in the next sections the \C/ itself.

We shall now give a brief account on Pusz and Woronicz's construction. Let $\V$ be a complex vector space and consider two positive, Hermitian quadratic forms defined on it
$$q, p\colon \V\times\V \to\CC.$$
The null space associated to these quadratic forms is defined as 
$$\mathcal{N} := \{ a \in \V \mid p(a,a) + q(a,a) = 0\}.$$
The set $\mathcal{N}$ is a linear subspace of $\V$. Using an approach analogous to the Gelfand–Naimark–Segal construction one can obtain a representation of the quadratic forms $p$ and $q$ on a Hilbert space. Indeed, the quotient $\V/\mathcal{N}$ defines a pre-Hilbert space with inner product 
$$\scalar{[a]}{[b]} = p(a,b) + q(a,b)\,,\quad [a], [b] \in \V/\mathcal{N},$$
and one can define $\H$ to be the closure of $\V/\mathcal{N}$ with respect to the norm induced by the inner product. Both quadratic forms on $\V$ define canonically quadratic forms on $\H$ which we will keep denoting with the same symbols, i.e.
$$\map[p]{\H\times\H}{\CC}{[a],[b]}{p(a,b),}$$
and analogously for $q$. Moreover, both quadratic forms are continuous with respect to the norm induced by $\scalar{\cdot}{\cdot}$ and by Riesz representation theorem they can be represented by two positive and bounded operators on $\H$. That is, there exist $P,Q\in \mathcal{B}_+(\H)$ such that for all $a,b\in\H$ one has that 
\begin{equation}\label{eq:compatible}
    p(a,b) = \scalar{a}{Pb} \quad \text{and} \quad q(a,b) = \scalar{a}{Qb}.
\end{equation}

Remarkably, these two operators commute. Notice that by the definition of the scalar product and of the operators $P$ and $Q$ one has that $P+Q = \mathbb{I}_\H$. As a pair of commuting, self-adjoint operators one can use them to define a functional calculus of quadratic forms. On their original article Pusz and Woronowicz consider the set of homogeneous measurable and locally bounded functions, but we will not pursue such generality here. Let $f\colon \R_+^2 \to \R$ be a homogeneous continuous function and let $f(P,Q)$ be the positive operator defined by the functional calculus on self-adjoint operators. Define a new quadratic form on $\V$ by the formula
\begin{equation}
    f_{[p,q]}(a,b) := \scalar{[a]}{f(P,Q)[b]}.
\end{equation}
The quadratic form so defined does not depend on the representation chosen. That is, consider that there is another Hilbert space $\H'$ and a pair of positive commuting operators on it $P'$ an $Q'$ representing the forms $p$ and $q$. The quadratic form defined by the functional calculus of the self-adjoint operators $P'$ and $ Q'$ gives rise to the same quadratic form on $\V$. We refer to \cite[Theorem 1.2]{pusz75} for the details of this proof. The particular representation obtained here by means of the GNS like construction is just a convenient way of providing such representation. A triple $(\H,P,Q)$ of a Hilbert space and two positive, bounded, commuting operators $P,Q \in\mathcal{B}_+(\H)$ that satisfy \eqref{eq:compatible} will be called a \textbf{compatible representation for the quadratic forms} $p$ and $q$. 

To exemplify how this construction works we are going to consider a particular example. Let $\hat{P}$ and $\hat{Q}$ be the positive, Hermitian operators on $\mathbb{C}^2$ given by 
$$\hat{P} = 
       \begin{pmatrix}
         2 & 1 \\ 1 & 2
       \end{pmatrix}
       \quad \hat{Q} = 
       \begin{pmatrix}
         2 & i \\ -i & 2
       \end{pmatrix}
       \quad
       [\hat{P},\hat{Q}] = 2i
       \begin{pmatrix}
         -1 & 0 \\ 0 & 1
       \end{pmatrix}
       $$
These two operators do not commute. They define two Hermitian, positive quadratic forms on $\mathbb{C}^2$. That is, for $a,b\in\mathbb{C}^2$
$$q(a,b) = a^+\hat{Q}b\;;\quad p(a,b) = a^+\hat{P}b,$$
where $a^+$ is the Hermitian transpose vector of $a$. The null space is in this case $\mathcal{N}=\{0\}$. This is so since the eigenvalues of the matrix $\hat{Q}+\hat{P}$ are $4+\sqrt{2}>0$. The Hilbert space in which one can define the representation of the forms is $\H = (\mathbb{C}^2, \scalar{\cdot}{\cdot})$ where the scalar product is defined as \hbox{$\scalar{a}{b} = a^+(\hat{P}+\hat{Q})b$}. Since there is no null space, we droped the notation for the equivalence classes.

The quadratic forms $q$ and $p$ can be represented in this Hilbert space. 
$$p(a,b) = a^+\hat{P}b=\scalar{a}{Pb}=a^+(\hat{P}+\hat{Q})Pb,$$
and similarly for $q$. Therefore, the matrices $P$ and $Q$  representing the forms $q$ and $p$  become respectively $P=(\hat{P}+\hat{Q})^{-1}\hat{P}$ and $Q=(\hat{P}+\hat{Q})^{-1}\hat{Q}$. Clearly one has that $P+Q = (\hat{P}+\hat{Q})^{-1}(\hat{P}+\hat{Q})=\mathbb{I}$. Notice also that although the matrices $\hat{P}$ and $\hat{Q}$ do not commute, the matrices $P$ and $Q$ do commute. 
$$[P,Q] = PQ +P^2 - P^2 - QP = P(P+Q) - (P+Q)P = 0.$$


\section{Interpolation of Quadratic Forms and Relative Entropy}\label{sec:interpolation}

In what follows we will consider a particular family of quadratic forms constructed using the functional calculus of quadratic forms.

\begin{definition}\label{def:interpolation}
    Let $p$ and $q$ be positive, Hermitian quadratic forms on a vector pace $\V$ and let $(\H,P,Q)$ be a compatible representation for them. Consider the family of functions $f^t\colon \R_+^2 \to \R$ defined by $f^t(x,y) = x^{1-t}y^t$ for $t\in [0,1]$. The \textbf{interpolation of the quadratic forms} $p$ and $q$ is the family $\{\gamma^t_{[p,q]}\}_{t\in [0,1]}$ of quadratic forms on $\V$ defined by
    $$\gamma^t_{[p,q]}(a,b) := \scalar{[a]}{P^{1-t}Q^t [b]}\,,\quad a,b \in \V.$$
\end{definition}

Notice that one has that $\gamma^0_{[p,q]} = p$ and $\gamma^1_{[p,q]} = q$ and the following interpolation property:
\begin{equation}\label{eq:interpolationprop}
    \gamma^t_{[\gamma^{t_1}_{[p,q]},\gamma^{t_2}_{[p,q]}]} = \gamma^{t'}_{[p,q]}\,,\quad t' = t_1(1-t) + t_2t. 
\end{equation}
This implies that successive interpolations of quadratic forms give rise to quadratic forms on the previous interpolation. We shall be interested in the following element of the interpolation.

\begin{definition}\label{def:quadraticmean}
    Let $\{\gamma^t_{[p,q]}\}_{t\in [0,1]}$ be the interpolation of the quadratic forms $p$ and $q$ on $\V$. The \textbf{geometric mean} of the quadratic forms $p$ and $q$ is the quadratic form $\gamma^{1/2}_{[p,q]}$.
\end{definition}

\begin{definition}
    Let $p$ and $q$ and $r$ be positive, Hermitian quadratic forms on $\V$. We will say that $r$ is \textbf{dominated by} $p$ and $q$ if for all $a,b \in\V$ one has that 
    $$|r(a,b)|^2 \leq p(a,a)q(b,b).$$
\end{definition}

\begin{theorem}{\cite[Theorem 2.1]{pusz75}}\label{thm:gmsup}
    Let $p$ and $q$ be quadratic forms on $\V$ and let $S$ be the space of quadratic forms dominated by $p$ and $q$.   The geometric mean of the forms $p$ and $q$ satisfies
    $$\gamma^{1/2}_{[p,q]}(a,a) = \sup_{r\in S} r(a,a)\,, \quad a\in \V.$$
\end{theorem}

\begin{proposition}\label{prop:1}
    Let, $p$, $p'$, $q$, and $q'$ be quadratic forms on $\V$ such that for all $a\in\V$ one has that $p(a,a)\geq p'(a,a)$ and $q(a,a) \geq q'(a,a)$. Then their respective geometric means satisfy
    $$\gamma^{1/2}_{[p,q]}(a,a) \geq \gamma^{1/2}_{[p',q']}(a,a)\,, \quad a\in \V. $$
\end{proposition}

\begin{proof}
From the definition of geometric mean and the Cauchy-Schwarz inequality it follows that the geometric mean of the quadratic forms $p'$ and $q'$ is dominated by $p'$ and $q'$ and therefore we have for all $a,b\in\V$ that
$$\gamma^{1/2}_{[p',q']}(a,b) \leq p'(a,a)q'(b,b) \leq  p(a,a)q(b,b).$$
This shows that $\gamma^{1/2}_{[p',q']}$ is dominated by $p$ and $q$ and therefore Theorem~\ref{thm:gmsup} implies that 
$$\gamma^{1/2}_{[p',q']}(a,a) \leq \gamma^{1/2}_{[p,q]}(a,a)\,,\quad a\in\V,$$
as we wanted to show.
\end{proof}

We will show now that this property extends to the full interpolation.

\begin{proposition}\label{prop:2}
    Let $\gamma^t_{[p,q]}$ and $\gamma^t_{[p',q']}$ be quadratic interpolations such that for all $a\in\V$ one has that $p(a,a)\geq p'(a,a)$ and $q(a,a) \geq q'(a,a)$. Then 
    $$\gamma^t_{[p,q]}(a,a) \geq \gamma^t_{[p',q']}(a,a)\,,\quad t\in[0,1],\;a\in\V.$$
\end{proposition}
\begin{proof}
    Having into account that $\gamma^0_{[p,q]} = p$ and $\gamma^1_{[p,q]} = q$ and respectively for the quadratic interpolation of $p'$ and $q'$, Proposition~\ref{prop:1} shows the result for $t=1/2$. From the interpolation property, see Eq.~\eqref{eq:interpolationprop}, and the definition of geometric mean we have that
    $$\gamma^{1/2}_{[\gamma^{t_1}_{[p,q]},\gamma^{t_2}_{[p,q]}]} = \gamma^{(t_1+t_2)/2}_{[p,q]}\,,\quad t_1,t_2 \in [0,1].$$ Therefore we can iteratively proof the result for the middle point of all the successive bisections of the interval, which is a dense subset. Notice that the functions $t\mapsto \gamma^t_{[p,q]}(a,a)$ and $t\mapsto \gamma^t_{[p',q']}(a,a)$, $a\in\V$, are continuous. Let $t_0 \in [0,1]$. For every $\epsilon>0$ there exists an interval $[s_1,s_2]$, containing $t_0$, such that $\gamma^{s_1}_{[p,q]}(a,a) \geq \gamma^{s_1}_{[p',q']}(a,a)$,  $|\gamma^{s_1}_{[p,q]}(a,a) - \gamma^{t_0}_{[p,q]}(a,a)| <\epsilon$ and $|\gamma^{s_1}_{[p',q']}(a,a) - \gamma^{t_0}_{[p',q']}(a,a)| <\epsilon$. Hence we have
    $$\gamma^{t_0}_{[p,q]}(a,a) + \epsilon \geq  \gamma^{s_1}_{[p,q]}(a,a) \geq \gamma^{s_1}_{[p',q']}(a,a) \geq \gamma^{t_0}_{[p',q']}(a,a) - \epsilon.$$
    Since this is true for every epsilon the proof is complete. 
    \end{proof}

Consider now that $\V$ and $\V'$ are two vector spaces and let $\Phi \colon \V' \to \V$ be a linear mapping. Given a quadratic form $p: \V \times \V \to \CC$ the linear map $\Phi$ induces a new quadratic form $\Phid p$ on $\V' $ by pull-back:
$$\map[\Phid p]{\V'\times \V'}{\CC}{a',b'}{p(\phi(a'),\phi(b'))}.$$

\begin{proposition}\label{prop:3}
    Let $\V$ and $\V'$ be two linear vector spaces and let $\Phi \colon \V' \to \V$ be a liner map. Let $p$ and $q$ be positive, Hermitian quadratic forms over $\V$. Then the pull-backs of the quadratic interpolations satisfy
    $$\Phid \gamma^t_{[p,q]}(a',a') \leq \gamma^t_{[\Phid p, \Phid q]}(a', a')\,,\quad a' \in \V'.$$
\end{proposition}

\begin{proof}
    Let $r$ be a positive, Hermitian quadratic form dominated by $p$ and $q$ and notice that the inequality $|r(a,b)|^2 \leq p(a,a)q(b,b)$, $a,b\in \V$ implies that $$|\Phid r(a',b')|^2 \leq \Phid p(a',a') \Phid q(b',b')\,, \quad a',b'\in \V'.$$ By the maximality property of the geometric mean, Theorem~\ref{thm:gmsup}, one has that for every such $r$
    $$\Phid r(a',a') \leq \gamma^{1/2}_{[\Phid p,\Phid q]}(a',a').$$
    In particular this holds for $r= \gamma^{1/2}_{[p,q]}$ and this proves the statement for $t=1/2$. The cases $t=0$ and $t=1$ are trivial. Repeating the final part of the argument in the proof of Proposition~\ref{prop:2} finishes the proof. \end{proof}

We are now ready to define the relative entropy between two states on a unital \C/. This is the definition originally introduced by A. Uhlmann \cite{uhlmann77}. In this section and for the rest of the article the generic vector spaces $\V$ of the previous section are going to be the provided by the vector space structure of the \C/. A state $\omega$ on a unital \C/ $\A$ is a real, normalised, positive, linear functional $\omega\colon \A\to\CC$, i.e., for $\lambda,\mu \in\CC$, $a,b \in \A$ and $e\in\A$ the identity element in the algebra
\begin{enumerate}
    \item $\omega(\lambda a + \mu b) = \lambda\omega(a) +\mu\omega(b)$
    \item $\omega(a^*a) \geq 0$ 
    \item $\omega(a^*) = \overline{\omega(a)}$
    \item $\omega(e) = 1$
\end{enumerate}
In terms of any state on a \C/, one can define two Hermitian, positive quadratic forms as follows.
$$\omega^R(a,b) = \omega(ba^*)\,,\quad a,b\in\A$$
$$\omega^L(a,b) = \omega(a^*b)\,,\quad a,b\in\A$$

\begin{definition}
Let $\omega$, $\nu$ be two states on a \C/ $\A$. And let $\gamma^t_{[\omega^R,\nu^L]}$ be the quadratic interpolation of the forms $\omega^R$ and $\nu^L$. The \textbf{relative entropy functional} between the states $\omega$, $\nu$ is defined by
$$S_{[\omega,\nu]}(a,b) = -\liminf_{t\to0^+}\frac{1}{t}\left( \gamma^t_{[\omega^R,\nu^L]}(a,b) - \omega^R(a,b) \right)\,,\quad a,b\in\A.$$
\end{definition}

\begin{definition}
Let $\omega$, $\nu$ be two states on a \C/ $\A$. The \textbf{relative entropy} between the states $\omega$ and $\nu$ is the evaluation on the identity of the relative entropy functional:
$$S[\omega,\nu] = S_{[\omega,\nu]}(e,e).$$
\end{definition}

Next we are going to provide an example with the connection of this definition of relative entropy with von Neumann's. Suppose that we are given two density matrices $\hat{\omega},\hat{\nu}:\mathbb{C}^n\to\mathbb{C}^n$, that is, positive, Hermitian matrices with trace one. For simplicity we are going to assume that they are strictly positive definite. They define respectively, by means of the trace, two linear functionals $\omega$, $\nu$ on $\A=M(\mathbb{C}^n)$. That is, for $a\in M(\mathbb{C}^n)$ 

$$\omega(a) = \operatorname{Tr}(\hat{\omega}a)\;,\quad\nu(a) = \operatorname{Tr}(\hat{\nu}a).$$
The quadratic forms associated to them are therefore:
$$\omega^R(a,b) = \operatorname{Tr}(\hat{\omega}ba^*)\;,\quad \nu^L(a,b) = \operatorname{Tr}(\hat{\nu}a^*b)\;,\quad a,b\in M(\mathbb{C}^n)$$

We are going to obtain a compatible representation, cf. Section~\ref{sec:functional_calculus}, for these quadratic forms on $\H=M(\mathbb{C}^n)$ with Hilbert-Schmidt scalar product,
$$\scalar{a}{b} = \operatorname{Tr}(a^*b).$$
We have that
$$\omega^R(a,b) = \operatorname{Tr}(\hat{\omega}ba^*)=\operatorname{Tr}(a^*\hat{\omega}b)=\scalar{a}{\hat{\omega}b} = \scalar{a}{L_{\omega}b},$$
where $L_\omega$ is the operator of left-multiplication on the algebra. Equivalently we have:
$$\nu^L(a,b) = \operatorname{Tr}(\hat{\nu}a^*b)=\operatorname{Tr}(a^*b\hat{\nu})=\scalar{a}{b\hat{\nu}} = \scalar{a}{R_{\nu}b},$$
where $R_\nu$ is the operator of right-multiplication on the algebra. Notice that $L_\omega$ and $R_\nu$ are self-adjoint, commuting operators and therefore define a compatible representation for the forms $\omega$ and $\nu$. The interpolation of quadratic forms becomes in this case
$$\gamma^t_{[\omega^R,\nu^L]}(a,b) = \operatorname{Tr}(a^* L_\omega^{1-t}R_\nu^tb) = \operatorname{Tr}(a^* \hat{\omega}^{1-t}b\hat{\nu}^t).$$

Applying now the definition of the relative entropy we get
$$S[\omega,\nu] = -\liminf_{t\to0} \frac{1}{t}\left( \operatorname{Tr}(\hat{\omega}^{1-t}\hat{\nu}^{t} - \operatorname{Tr}(\hat{\omega})) \right).$$
This limit is minus the derivative of the first summand evaluated at zero, if it exists. A straightforward calculation shows that 
$$S[\omega,\nu] = -\operatorname{Tr}(\hat{\omega}\ln\hat{\nu}) + \operatorname{Tr}(\hat{\omega}\ln\hat{\omega}),$$
which is the expression of von Neumann's relative entropy.


\section{Proof of the monotonicity of the Relative Entropy}\label{sec:Monotonicity}

We are going to prove in this section the monotonicity of the relative entropy under completely positive and unital (with trace preserving dual) maps. In fact, the prove is slightly more general, and proves monotonicity under Schwarz maps. Let ${\Phi}:\tilde{\A} \to \A$ be a linear map between $*$-algebras with the following properties:

\begin{itemize}
  \item $\Phi(a^*) = \Phi(a)^*$
  \item $\Phi(a^*)\Phi(a)\leq\Phi(a^*a)$
\end{itemize}

Given a state defined on $\A$, $\omega:\A\to\mathbb{C}$, the map $\Phi$ induces a state in $\tilde{A}$ in the following way: 

$$\map[\omega_{\Phi}]{\tilde{A}}{\mathbb{C}}{\tilde{a}}{\omega(\Phi(\tilde{a}))}$$
This state satisfies:
$$\omega_\Phi(\ta^*\ta) = \omega(\Phi(\ta^*\ta)) \geq \omega(\Phi(\ta)^*\Phi(\ta))\geq 0,$$
which proves that a Schwarz map is a positive map. It is a known result that a completely positive map that is trace non-increasing, i.e. $e_{{\A}}\geq \Phi(e_{\tilde{\A}})$, is Schwarz. We will prove the monotonicity of the relative entropy under Schwarz maps that are trace preserving.

\begin{theorem}[Monotonicity of the relative entropy]
Let $\tilde{\A}$ and $\A$ be unital $C^*$-algebras and $\Phi:\tilde{\A}\to\A$ be a unital (with trace preserving dual) Schwarz map. Let $\omega,\nu$ be positive, linear functionals on $\A$. Then 
$$S[\omega,\nu] \geq S[\omega_\Phi,\nu_\Phi]$$
\end{theorem}

Before we proceed with the proof it is important to remark that \hbox{$\Phi^\dagger\omega^R\neq \omega_\Phi^R$}. Indeed, on one hand we have for $\ta\in\tilde{\A}$ that
$$\Phi^\dagger\omega^R(\ta,\ta) = \omega^R(\Phi(\ta),\Phi(\ta)) = \omega(\Phi(a)\Phi(a)^*).$$
On the other hand
$$\omega_\Phi^R(\ta,\ta) = \omega_\Phi(\ta\ta^*)=\omega(\Phi(aa^*)).$$
Hence, only in the case that $\Phi$ is a $*$-homomorphism we get the equality.

\begin{proof}
Let $e$ and $\te$ be the identity elements in $\A$ and $\tilde{\A}$ respectively. By definition we have that 
$$S[\omega,\nu]  = -\liminf_{t\to0^+}\frac{1}{t}\left( \gamma^t_{[\omega^R,\nu^L]}(e,e) - \omega^R(e,e) \right)$$
$$S[\omega_\Phi,\nu_\Phi]  = -\liminf_{t\to0^+}\frac{1}{t}\left( \gamma^t_{[\omega_\Phi^R,\nu_\Phi^L]}(\te,\te) - \omega_\Phi^R(\te,\te) \right)$$
Since the map is unital we have the following
$$\omega^R(e,e) = \omega(ee^*) = \omega(e) = \omega(\Phi(\te)) =  \omega(\Phi(\te\te^*)) = \omega_\Phi(\te\te^*) = \omega_\Phi^R(\te,\te).$$

Therefore, it is enough to show that 
$$\gamma^t_{[\omega^R,\nu^L]}(e,e) = \gamma^t_{[\omega^R,\nu^L]}(\Phi(\te),\Phi(\te)) \leq \gamma^t_{[\omega_\Phi^R,\nu_\Phi^L]}(\te,\te).$$
In particular, this will hold if 
$$\gamma^t_{[\omega^R,\nu^L]}(\Phi(\ta),\Phi(\ta)) \leq \gamma^t_{[\omega_\Phi^R,\nu_\Phi^L]}(\ta,\ta)\;,\quad \forall \ta \in \tilde{\A}.$$

Using the Schwarz property of the map $\Phi$ we get that 
\begin{align*}
\Phi^\dagger\omega^R(\ta,\ta) 
    & = \omega^R(\Phi(\ta),\Phi(\ta)) \\ 
    & = \omega(\Phi(\ta)\Phi(\ta)^*)\\
    & \leq \omega(\Phi(\ta\ta^*))\\
    & = \omega_\Phi(\ta\ta^*) \\
    & = \omega_\Phi^R(\ta,\ta),
\end{align*}
and equivalently for $\nu^L$.

Now one can use Proposition~\ref{prop:2} and Proposition~\ref{prop:3} to get
$$\gamma^t_{[\omega^R,\nu^L]}(\Phi(\ta),\Phi(\ta)) = \Phi^\dagger\gamma^t_{[\omega^R,\nu^L]}(\ta,\ta) \leq \gamma^t_{[\Phid \omega^R, \Phid \nu^L]}(\ta,\ta) \leq \gamma^t_{[\omega_\Phi^R, \nu_\Phi^L]}(\ta,\ta),$$
which completes the proof.
\end{proof}

The conditions of the Theorem apply in the particular case that $\tilde{\A}\subset\A$ is a unital $*$-subalgebra and $\Phi:\tilde{\A}\to\A$ is the injection map. Since the injection is a $*$-homomorphism and unital, it is a completely positive and trace preserving map. This means that the restriction of the observables to a subalgebra decreases the relative entropy.

\section*{Acknowledgements}
{This work was partially supported by the “Ministerio de Ciencia e Innovación” Research Project PID2020-117477GB-I00, by the Madrid Government (Comunidad de Madrid-Spain) under the Multiannual Agreement with UC3M in the line of “Research Funds for Beatriz Galindo Fellowships” (C\&QIG-BG-CM-UC3M), and in the context of the V PRICIT (Regional Programme of Research and Technological Innovation), by the QUITEMAD Project P2018/TCS-4342 funded by Madrid Government (Comunidad de Madrid-Spain) and by the Severo Ochoa Programme for Centers of Excellence in R\&D” (CEX2019-000904-S).\\ email: jmppardo@math.uc3m.es}


\end{document}